\documentclass[twocolumn,a4paper]{article}

\usepackage[a4paper, total={7in, 9.5in}]{geometry}
\usepackage{url}
\usepackage[T1]{fontenc}
\usepackage{amsmath, amssymb, amsthm, bm}
\usepackage{algorithm}
\usepackage{algpseudocode}
\usepackage{graphicx,dblfloatfix}
\usepackage{xspace}
\usepackage{subcaption}
\usepackage{array}
\usepackage{xcolor}
\usepackage{calc}
\usepackage{mathtools}
\usepackage{booktabs,tabularx}
\usepackage{kpfonts}

\usepackage{tikz,pgfplots}
\usetikzlibrary{plotmarks}
\pgfplotsset{compat=newest}

\usepackage[sort&compress]{natbib}

\newtheorem{theorem}{Theorem}[section]

\newcommand{\eg}{\textit{e.g.}\xspace}
\newcommand{\ie}{\textit{i.e.}\xspace}

\newcommand{\Reals}{\mathbb{R}}

\newcommand{\N}{\mathcal{N}}
\newcommand{\GP}{\mathcal{GP}}
\newcommand{\weight}{w}

\newcommand{\Prb}[1]{\mathbb{P}({#1})}
\newcommand{\Transp}{^\mathsf{T}}

\newcommand{\ii}{{(i)}}
\newcommand{\ji}{{(j)}}

\newcommand{\anc}{a}
\newcommand{\epdf}[1]{\mathcal{N}({#1}\mid \vect{y}_t, \mat{R})}
\newcommand{\qpdf}[2]{\mathcal{N}({#1}\mid {#2},\mat{Q})}
\newcommand{\qdist}[1]{\mathcal{N}({#1},\mat{Q})}
\newcommand{\TO}{{\bf to }}

\newcommand{\imag}{\mathrm{i}}
\newcommand{\diff}{\,\mathrm{d}}

\newcommand{\vect}[1]{\mathbf{#1}}
\newcommand{\vectb}[1]{\bm{#1}}
\newcommand{\mat}[1]{\mathbf{#1}}
\newcommand{\matb}[1]{\bm{#1}}

\newcommand{\coverTitle}{\textbf{Computationally Efficient Bayesian Learning of Gaussian Process State Space Models}}
\newcommand{\coverAuthors}{Andreas Svensson, Arno Solin, Simo S\"{a}rkk\"{a} and Thomas B. Sch\"{o}n}

\title{Computationally Efficient Bayesian Learning \\of Gaussian Process State Space Models}
\author{Andreas Svensson\\Uppsala University \and Arno Solin\\Aalto University \and Simo S\"{a}rkk\"{a}\\Aalto University \and Thomas B. Sch\"{o}n\\Uppsala University}

\date{}

\begin{document}

\begin{titlepage}
	\begin{center}
		{\large \em Technical report}
		\vspace*{2.5cm}

		{\huge \bfseries \coverTitle  \\[0.4cm]}
		
		{\Large \coverAuthors \\[2cm]}
		
		\renewcommand\labelitemi{\color{red}\large$\bullet$}
		\begin{itemize}
			\item {\Large \textbf{Please cite this version:}} \\[0.4cm]
			\large
			\coverAuthors. \coverTitle. In \textit{Proceedings of the
				$\mathit{19}$\textsuperscript{th} International Conference on 
				Artificial Intelligence and Statistics (AISTATS)},
			Cadiz, Spain, 2016.
			
			{\tiny
			\begin{verbatim}
			@InProceedings{SvenssonSSS2016,
			Title     = {Computationally efficient {B}ayesian learning of {G}aussian process state space models},
			Author    = {Svensson, Andreas and Solin, Arno and S{\"a}rkk{\"a}, Simo and Sch\"{o}n, Thomas B.},
			Booktitle = {Proceedings of 19\textsuperscript{th} International Conference on Artificial Intelligence and Statistics (AISTATS)},
			Year      = {2016},
			Address   = {Cadiz, Spain},
			Month     = {May},
			}
			
			\end{verbatim}
		}
			
		\end{itemize}
		\vfill
		
\begin{abstract}
	Gaussian processes allow for flexible specification of prior assumptions of unknown dynamics in state space models. We present a procedure for efficient Bayesian learning in Gaussian process state space models, where the representation is formed by projecting the problem onto a set of approximate eigenfunctions derived from the prior covariance structure. Learning under this family of models can be conducted using a carefully crafted particle MCMC algorithm. This scheme is computationally  efficient and yet allows for a fully Bayesian treatment of the problem. Compared to conventional system identification tools or existing learning methods, we show competitive performance and reliable quantification of uncertainties in the model.
\end{abstract}

		\vfill
	\end{center}
\end{titlepage}

\maketitle

\begin{abstract}
  Gaussian processes allow for flexible specification of prior assumptions of unknown dynamics in state space models. We present a procedure for efficient Bayesian learning in Gaussian process state space models, where the representation is formed by projecting the problem onto a set of approximate eigenfunctions derived from the prior covariance structure. Learning under this family of models can be conducted using a carefully crafted particle MCMC algorithm. This scheme is computationally  efficient and yet allows for a fully Bayesian treatment of the problem. Compared to conventional system identification tools or existing learning methods, we show competitive performance and reliable quantification of uncertainties in the model.
\end{abstract}

\section{INTRODUCTION}

Gaussian processes (GPs, \citealt{RW:2006}) have been proven to be powerful probabilistic non-parametric modeling tools for \emph{static} nonlinear functions. However, many real-world applications, such as control, target tracking, and time-series analysis are tackling problems with nonlinear \emph{dynamical} behavior. The use of GPs in modeling nonlinear dynamical systems is an emerging topic, with many strong contributions during the recent years, for example the work by \citet{TDR:2010}, \cite{FLS+:2013, FLS+:2014, FCR:2014} and \cite{MDD+:2015}. The aim of this paper is to advance the state-of-the-art in Bayesian inference on Gaussian process state space models (GP-SSMs). As we will detail, a GP-SSM is a state space model, using a GP as its state transition function. Thus, the GP-SSM is not a GP itself, but a state space model (\ie, a dynamical system). Overviews of GP-SSMs are given by, \eg, \citet{McHutchon:2014} and \citet{Frigola:2015}.

We provide a novel reduced-rank model formulation of the GP-SSM with good convergence properties both in theory and practice. The advantage with our approach over the variational approach by \cite{FCR:2014}, as well as other inducing-point-based approaches, is that our approach attempts to approximate the optimal Karhunen--Loeve eigenbasis for the reduced-rank approximation instead of using the sub-optimal Nystr\"om approximation which implicitly is the underlying approximation in all inducing point methods. Because of this we do not need to resort to variational approximations, but we can instead perform the Bayesian computations in full. By utilizing the structure of the reduced-rank model, we construct a computationally efficient linear-time-complexity MCMC-based algorithm for learning in the proposed GP-SSM model, which we  demonstrate and evaluate on several challenging examples. We also provide a proof of convergence of the reduced-rank GP-SSM to a full GP-SSM (in the supplementary material).

\begin{figure*}[t!]
	\centering
	\begin{subfigure}[b]{0.49\textwidth}
		\centering%
		\includegraphics{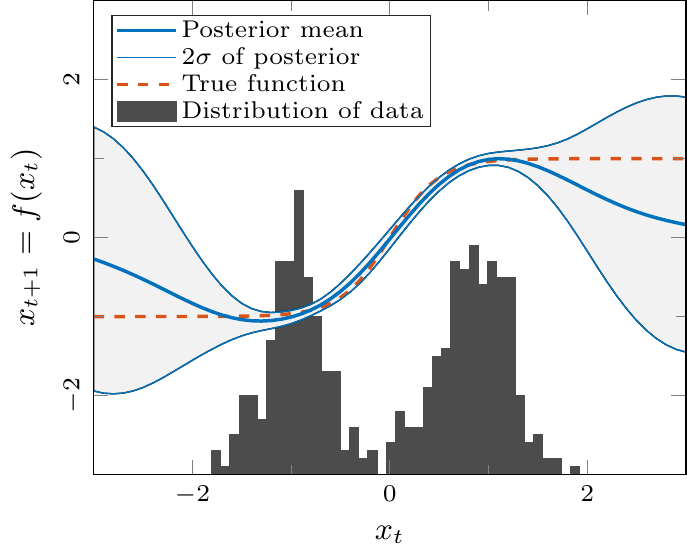}
		\caption{The learned model}
		\label{fig:simple-posterior}
	\end{subfigure}
	\hspace*{\fill}
	\begin{subfigure}[b]{0.49\textwidth}
		\centering%
		\includegraphics{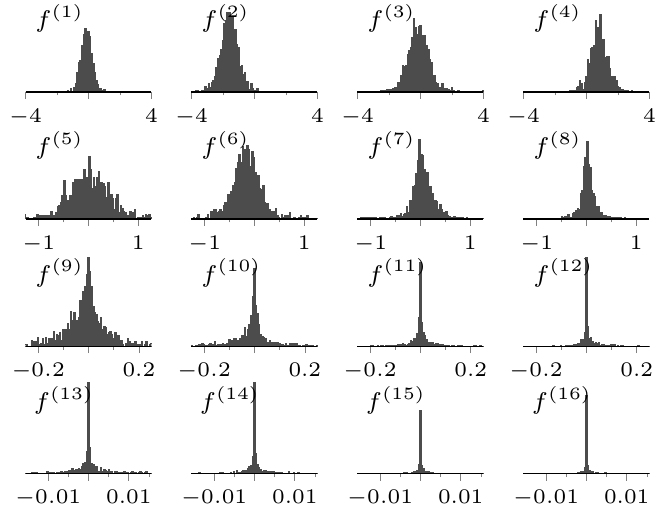}
		\vspace*{12pt}
		\caption{The posterior weights $f^{(i)}$}
		\label{fig:simple-histogram}
	\end{subfigure}
	\hspace*{\fill}
	\caption{An example illustrating how the GP-SSMs handle uncertainty. (a)~The learned model from data $y_{1:T}$. The bars show where the data is located in the state space, \ie, what part of the model is excited in the data set, affecting the posterior uncertainty in the learned model. (b)~Our approach relies on a basis function expansion of $f$, and learning $f$ amounts to finding the posterior distribution of the weights $f^{(i)}$ depicted by the histograms.}
	\label{fig:simple}
\end{figure*}

GP-SSMs are a general class of models defining a dynamical system for $t = 1, 2, \dots, T$ given by
\begin{subequations}\label{eq:ssm}
  \begin{align}
    \vect{x}_{t+1} &= \vect{f}(\vect{x}_t) + \vect{w}_t, \nonumber \\
    &\quad\quad\text{with}~~ \vect{f}(\vect{x}) 
    \sim \GP(\vect{0},\vectb{\kappa}_{\theta,f}(\vect{x},\vect{x}')), \label{eq:ssmx}\\
    \vect{y}_t     &= \vect{g}(\vect{x}_t) + \vect{e}_t, \nonumber \\
    &\quad\quad\text{with}~~ \vect{g}(\vect{x}) 
    \sim \GP(\vect{0},\vectb{\kappa}_{\theta,g}(\vect{x},\vect{x}')), \label{eq:ssmy}
  \end{align}
\end{subequations}
where the noise terms $\vect{w}_t$ and $\vect{e}_t$ are i.i.d.\ Gaussian, $\vect{w}_t \sim \N(\vect{0},\mat{Q})$ and $\vect{e}_t \sim \N(\vect{0},\mat{R})$. The latent state $\vect{x}_t \in \Reals^{n_x}$ is observed via the measurements $\vect{y}_t \in \Reals^{n_y}$. The key feature of this model is the nonlinear transformations $\vect{f} : \Reals^{n_x} \to \Reals^{n_x}$ and $\vect{g} : \Reals^{n_x} \to \Reals^{n_y}$ which are not known explicitly and do not adhere to any specific parametrization. The model functions $\vect{f}$ and $\vect{g}$ are assumed to be realizations from a Gaussian process prior over $\Reals^{n_x}$ with a given covariance function $\vectb{\kappa}_\theta(\vect{x},\vect{x}')$ subject to some hyperparameters $\vectb{\theta}$. Learning of this model, which we will tackle, amounts to inferring the posterior distribution of $\vect{f}$, $\vect{g}$, $\vect{Q}$, $\vect{R}$, and $\vectb{\theta}$ given a set of (noisy) observations $\vect{y}_{1:T} \triangleq \{\vect{y}_i\}_{i=1}^T$.

The strength of including the GP in \eqref{eq:ssm} is its ability to systematically model \emph{uncertainty}---not only uncertainty originating from stochastic noise within the system, but also uncertainty inherited from data, such as few measurements or poor excitation of the dynamics in certain regions of the state space. An example of this is given by Figure~\ref{fig:simple}, where we learn the posterior \emph{distribution} of the unknown function $\vect{f}(\cdot)$ in a GP-SSM (see Sec.~\ref{sec:examples} for details). An inspiring real-world example on how such probabilistic information can be utilized for simultaneous learning and control is given by \cite{DFR:2015}.

Non-probabilistic methods for modeling nonlinear dynamical systems include learning of state space models using a basis function expansion \citep{GR:1999}, but also nonlinear extensions of AR(MA) and GARCH models from the time-series analysis literature \citep{Tsay:2010}, as well as nonlinear extensions of ARX and state space models from the system identification literature \citep{SZL+:1995,Ljung:1999}. In particular, nonlinear ARX models are now a standard tool for the system identification engineer \citep{Mathworks:2015}. For probabilistic modeling, the latent force model \citep{ALL:2009} presents one approach for modeling dynamical phenomena using GPs by encoding \textit{a~priori} known dynamics within the construction of the GP. Another approach is the Gaussian process dynamical model \citep{WFH:2008}, where a GP is used to model the nonlinear function within an SSM, that is, a GP-SSM. 
However, the work by \citet{WFH:2008} is, as opposed to this paper, mostly focused around the problem setting when $n_y \gg n_x$. That is also the focus for the further development by \citet{DTL:2011}, where the EM algorithm for learning is replaced by a variational approach.

State space filtering and smoothing in GP-SSMs has been tackled before (\eg, \citealt{DTH+:2012, DM:2012}), and recent interest has been in learning GP-SSMs  \citep{TDR:2010, FLS+:2013, FLS+:2014, FCR:2014}. An inherent problem in learning the GP-SSM is the entangled relationship between the states $\vect{x}_t$ and the nonlinear function $\vect{f}(\cdot)$. Two different approaches have been proposed in the literature: In the first approach the GP is represented by a parametrized form (\citeauthor{TDR:2010} use a pseudo-training data set, akin to the inducing inputs by \citealt{FCR:2014}, whereas we will employ a basis function expansion). The second approach (used by \citealt{FLS+:2013, FLS+:2014}) is handling the nonlinear function implicitly by marginalizing it out. Concerning learning, \cite{TDR:2010} and \cite{FLS+:2014} use an EM-based procedure, whereas we and \cite{FLS+:2013} use an MCMC algorithm.

The main bottleneck prohibiting the use in practice of some of the previously proposed GP-SSMs methods is the computational load. For example, the training of a one-dimensional system using $T = 500$ data points (\ie, a fairly small example) is in the magnitude of several hours for the solution by \cite{FLS+:2013}. Akin to \cite{FCR:2014}, our proposed method will typically handle such an example within minutes, or even less. To reduce the computational load, \cite{FCR:2014} suggests variational sparse GP techniques to approximate the solution. Our approach, however, is using the reduced-rank GP approximation by \cite{SS:2014}, which is a disparate solution with different properties. The reduced-rank GP approximation enjoys favorable theoretical properties, and we can prove convergence to a non-approximated GP-SSM.

The outline of the paper is as follows: In Section~\ref{sec:rrgpssm} we will introduce reduced-rank Gaussian process state space models by making use of the representation of GPs via basis functions corresponding to the prior covariance structure \citep{SS:2014}, a theoretically well-supported approximation significantly reducing the computational load. In Section~\ref{sec:learning} we will develop an algorithm for learning reduced-rank Gaussian process state space models by using recent MCMC methods \citep{LJS:2014, WSL+:2012}. We will also demonstrate it on synthetic as well as real data examples in Section~\ref{sec:examples}, and finally discuss the contribution and further extensions in Section~\ref{sec:discussion}.

\pagebreak

\section{REDUCED-RANK GP-SSMs}
\label{sec:rrgpssm}

We use GPs as flexible priors in Bayesian learning of the state space model. The covariance function $\kappa(\vect{x},\vect{x}')$ encodes the prior assumptions of the model functions, thus representing the best belief of the behavior of the nonlinear transformations. In the following we present an approach for parametrizing this model in terms of an $m$-rank truncation of a basis function expansion as presented by \citet{SS:2014}. Related ideas have also been proposed by, for example, \citet{LQR+:2010}.

Provided that the covariance function is stationary (homogeneous, \ie\ $\kappa(\vect{x}-\vect{x}') \triangleq \kappa(\vect{x},\vect{x}')$), the covariance function can be equivalently represented in terms of the spectral density $S(\vectb{\omega})$. This Fourier duality is known as the \emph{Wiener--Khintchin theorem}, which we parametrize as: $S(\vectb{\omega}) = \int \kappa(\vect{r}) \, \exp({-\imag \, \vectb{\omega}\Transp \vect{r}}) \diff \vect{r}$. We employ the relation presented by \citet{SS:2014} to approximate the covariance operator corresponding to $\kappa(\cdot)$. This operator is a pseudo-differential operator, which we approximate by a series of differential operators, namely Laplace operators $\nabla^2$. In the isotropic case, the approximation of the covariance function is given most concisely in the following form:
\begin{equation} \label{eq:approx}
  \kappa_\theta(\vect{x},\vect{x}') \approx \sum_{j=1}^m S_\theta(\lambda_j) \, \phi^{(j)}(\vect{x})\,\phi^{(j)}(\vect{x}'),
\end{equation}
where $S_\theta(\cdot)$ is the spectral density function of $\kappa_\theta(\cdot)$, and $\lambda_j$ and $\phi^{(j)}$ are the Laplace operator eigenvalues and eigenfunctions solved for the domain $\Omega \ni \vect{x}$. See \cite{SS:2014} for a detailed derivation and convergence proofs.

The key feature in the Hilbert space approximation \eqref{eq:approx} is that $\lambda_j$ and $\phi^{(j)}$ are independent of the hyperparameters $\vectb{\theta}$, and it is only the spectral density that depends on $\vectb{\theta}$. Equation~\eqref{eq:approx} is a direct approximation of the eigendecomposition of the Gram matrix (\eg, \citealt{RW:2006}), and it can be interpreted as an optimal parametric expansion with respect to the given covariance function in the GP prior.

In terms of a basis function expansion, this can be expressed as
\begin{align}\label{eq:f_prior}
  f(\vect{x}) \sim \GP(0,\kappa(\vect{x},\vect{x}'))
  &\quad \Leftrightarrow \quad 
  f(\vect{x}) \approx \sum_{j=1}^m f^{(j)} \phi^{(j)}(\vect{x}), 
\end{align}
where $f^{(j)} \sim \N(0,S(\lambda_j))$. In the case $n_x>1$, this formulation does allow for non-zero covariance between different components of the state space.
We can now formulate a reduced-rank GP-SSM, corresponding to \eqref{eq:ssmx}, as
\begin{equation}\label{eq:rrgpss-expanded}
	\vect{x}_{t+1} =
	\underbrace{\begin{bmatrix} f_{1}^{(1)} & \hdots & f_{1}^{(m)} \\ \vdots & & \vdots \\ f_{n_x}^{(1)} & \hdots & f_{n_x}^{(m)}  \end{bmatrix}}_{\mat{A}}
	\underbrace{\begin{bmatrix} \phi^{(1)}(\vect{x}_t) \vphantom{f_{1}^{(m)}} \\ \vdots \\ \phi^{(m)}(\vect{x}_t) \vphantom{f_{n_x}^{(m)}} \end{bmatrix}}_{\matb{\Phi}(\vect{x}_t)} + 
	\vect{w}_t,
\end{equation}
and similarly for \eqref{eq:ssmy}. Henceforth we will consider a reduced-rank GP-SSM,
\begin{subequations}\label{eq:rrgpssm}
	\begin{align}
		\vect{x}_{t+1} &= \mat{A}\Phi(\vect{x}_t) + \vect{w}_t, \\
		\vect{y}_{t} &= \mat{C}\Phi(\vect{x}_t) + \vect{e}_t,
	\end{align}
\end{subequations}
where $\mat{A}$ and $\mat{C}$ are matrices of weights with priors for each element as described by \eqref{eq:f_prior}.

\section{LEARNING GP-SSMs}
\label{sec:learning}
\begin{algorithm*}[t!]
	\caption{Learning of reduced-rank GP-SSMs.}
	\begin{algorithmic}[1]
		\Require Data $y_{1:T}$, priors on $\mat{A}$, $\mat{Q}$ and $\vectb{\theta}$.
		\Ensure $K$ MCMC-samples with $p(\vect{x}_{1:T},\mat{Q},\mat{A},\vectb{\theta}\mid\vect{y}_{1:T})$ as invariant distribution.
		\State Sample initial $\vect{x}_{1:T}[0],\mat{Q}[0],\mat{A}[0],\vectb{\theta}[0]$.
		\For{$k = 0$ to $K$}
		\State\label{alg:gibbs:x}Sample ${\vect{x}_{1:T}[k+1]}\boldsymbol{\:\big\vert\:} \mathrlap{\mat{Q}[k], \mat{A}[k], \vectb{\theta}[k]}$\phantom{$\vect{x}_{1:T}[k+1], \mat{Q}[k+1], \mat{A}[k+1]$} by Algorithm~\ref{alg:CPFAS}.
		\State\label{alg:gibbs:q}Sample \phantom{$\vect{x}_{1:T}[k+1]$}$\mathllap{\mat{Q}[k+1]}\boldsymbol{\:\big\vert\:} \mathrlap{\mat{A}[k], \vectb{\theta}[k], \vect{x}_{1:T}[k+1]}$\phantom{$\vect{x}_{1:T}[k+1], \mat{Q}[k+1], \mat{A}[k+1]$} according to \eqref{eq:Qdist}.
		\State\label{alg:gibbs:a}Sample \phantom{$\vect{x}_{1:T}[k+1]$}$\mathllap{\mat{A}[k+1]}\boldsymbol{\:\big\vert\:} \mathrlap{\vectb{\theta}[k], \vect{x}_{1:T}[k+1], \mat{Q}[k+1]}$\phantom{$\vect{x}_{1:T}[k+1], \mat{Q}[k+1], \mat{A}[k+1]$} according to \eqref{eq:Adist}.
		\State\label{alg:gibbs:eta}Sample \phantom{$\vect{x}_{1:T}[k+1]$}$\mathllap{\vectb{\theta}[k+1]}\boldsymbol{\:\big\vert\:} \vect{x}_{1:T}[k+1], \mat{Q}[k+1], \mat{A}[k+1]$ by using MH (Section~\ref{sec:sample_hyperparams}).
		\EndFor
	\end{algorithmic}
	\label{alg:gibbs}
\end{algorithm*}
Learning in reduced-rank Gaussian process state space models \eqref{eq:rrgpssm} from $\vect{y}_{1:T}$ amounts to inferring the posterior distribution of $\mat{A}$, $\mat{C}$, $\mat{Q}$, $\mat{R}$, and the hyperparameters $\vectb{\theta}$. For clarity in the presentation, we will focus on inferring the dynamics, and assume the observation model ($\vect{g}(\cdot)$ and $\mat{R}$) to be known \textit{a~priori}. However, the extension to an unknown observation model---as well as exogenous input signals---follows in the same fashion, and will be demonstrated in the numerical examples.

To infer the sought distributions, we will use a blocked Gibbs sampler outlined in Algorithm~\ref{alg:gibbs}. Although involving sequential Monte Carlo (SMC) for inference in state space, the validity of this approach is \emph{not} relying on asymptotics ($N \to \infty$) in the SMC algorithm, thanks to recent particle MCMC methods \citep{ADH:2010,LJS:2014}.

It is possible to learn \eqref{eq:rrgpssm} under different assumptions on what is known. We will focus on the general (and in many cases realistic) setting where the distributions of $\mat{A}$, $\mat{Q}$ and $\vectb{\theta}$ are all unknown. In cases when $\mat{Q}$ or $\vectb{\theta}$ are known \textit{a~priori}, the presented scheme is straightforward to adapt. To be able to infer the posterior distribution of $\mat{Q}$ and $\vectb{\theta}$, we make the additional prior assumptions:
\begin{align} \label{eq:Qprior}
	\mat{Q} &\sim \mathcal{IW}(\ell_Q, \matb{\Lambda}_Q), &
	\vectb{\theta} &\sim p(\vectb{\theta}), 
\end{align}
where $\mathcal{IW}$ denotes the Inverse Wishart distribution. For brevity, we will omit the problem of finding the unknown initial distribution $p(\vect{x}_1)$. It is possible to treat this rigorously akin to $\vectb{\theta}$, but it is of minor importance in most practical situations. We will now in Section~\ref{sec:sample_x}--\ref{sec:sample_hyperparams} explain the four main steps~\ref{alg:gibbs:x}--\ref{alg:gibbs:eta} in Algorithm~\ref{alg:gibbs}.

\subsection{Sampling in State Space with SMC}
\label{sec:sample_x}

SMC methods \citep{DJ:2011} are a family of techniques developed around the problem of inferring the posterior state distribution in SSMs. SMC can be seen as a sequential application of importance sampling along the sequence of distributions $\ldots, p(\vect{x}_{t-1} \mid \vect{y}_{1:t-1}), p(\vect{x}_t \mid \vect{y}_{1:t}), \ldots$ with a resampling procedure to avoid sample depletion.

To sample the state space trajectory $\vect{x}_{1:T}$, conditional on a model $\mat{A}$, $\mat{Q}$ and data $\vect{y}_{1:T}$, we employ a conditional particle filter with ancestor sampling, forming a particle Gibbs Markov kernel Algorithm~\ref{alg:CPFAS} (PGAS, \citealt{LJS:2014}). PGAS can be thought of as an SMC algorithm for finding the so-called smoothing distribution $p(\vect{x}_{1:T}\mid \mat{A}, \mat{Q},\vect{y}_{1:T})$ to be used within an MCMC procedure.

\begin{algorithm}[!b]
	\caption{Particle Gibbs Markov kernel.}
	\label{alg:CPFAS}
	\begin{algorithmic}[1]\scriptsize
		\Require Trajectory $\vect{x}_{1:T}[k]$, number of particles $N$
		\Ensure Trajectory $\vect{x}_{1:T}[k+1]$
		\State Sample $\vect{x}_1^\ii \sim p(\vect{x}_1)$, for $i = 1,\dots,N-1$.
		\State Set $\vect{x}_1^N = \vect{x}_1[k]$.
		\State \textbf{For} $t = 1$ \TO $T$
		\State ~Set $w_t^\ii = p(\vect{y}_t\mid\vect{x}_t^\ii) =  \epdf{\vect{g}(\vect{x}_t^\ii)}$, for $i = 1, \dots, N$.
		\State ~\label{alg:CPFAS:resampling}Sample $\anc_t^\ii$ with $\Prb{\anc_t^\ii=j} \propto w_{t}^\ji$, for $i = 1, \dots, N-1$.
		\State ~Sample $\vect{x}_{t+1}^\ii \sim \qdist{\vect{f}(\vect{x}_t^{\anc_t^\ii})}$, for $i = 1, \dots, N-1$.
		\State ~Set $\vect{x}_{t+1}^N = \vect{x}_{t+1}[k]$.
		\State ~Sample $\anc_t^N$ with $\Prb{\anc_t^N = j} \propto$
		\Statex ~$w_{t}^\ji p(\vect{x}_{t+1}^N\mid\vect{x}_t^\ji) = w_{t}^\ji\qpdf{\vect{x}_{t+1}^N}{\vect{f}(\vect{x}_t^\ji)}$.
		\State ~Set $\vect{x}_{1:t+1}^\ii = \{\vect{x}_{1:t}^{\anc_t^\ii},\vect{x}_{t+1}^\ii\}$, for $i = 1, \dots, N$.
		\State \textbf{End for}
		\State Sample $J$ with $\Prb{J=i}\propto \weight_T^\ii$ and set $\vect{x}_{1:T}[k+1] = \vect{x}_{1:T}^J$.
	\end{algorithmic}
\end{algorithm}

\subsection{Sampling of Covariances and Weights}
\label{sec:sample_hyp}
The sampling of the weights $\mat{A}$ and the noise covariance $\mat{Q}$, conditioned on $\vect{x}_{1:T}$ and $\vectb{\theta}$, can be done exactly, by following the procedure of \citet{WSL+:2012}. With the priors \eqref{eq:f_prior} and \eqref{eq:Qprior}, the joint prior of $\mat{A}$ and $\vect{Q}$ can be written using the Matrix Normal Inverse Wishart (MNIW) distribution as
\begin{equation}\label{eq:A_prior}
  p(\mat{A},\vect{Q}) = \mathcal{MNIW}(\mat{A},\vect{Q}\mid\mat{0},\mat{V},\ell_Q,\matb{\Lambda}_Q).
\end{equation}
Details on the parametrization of the MNIW distribution we use is available in the supplementary material, and it is given by the hierarchical model $p(\mat{Q}) = \mathcal{IW}(\vect{Q}\mid\ell_Q,\matb{\Lambda}_Q)$ and $p(\mat{A}\mid\vect{Q}) =\mathcal{MN}(\mat{A}\mid\mat{0},\vect{Q},\mat{V})$. For our problem, the most important is the second argument, the inverse row covariance $\mat{V}$, a square matrix with the inverse spectral density of the covariance function as its diagonal entries:
\begin{equation}\label{eq:GP_prior}
  \mat{V} = \mathrm{diag}\left( [S^{-1}(\lambda_1) ~ \cdots ~ S^{-1}(\lambda_m) ]\right).
\end{equation}
This is how the prior from \eqref{eq:f_prior} enters the formulation. (Note that the marginal variance of each element in $\mat{A}$ is also scaled $\mat{Q}$, and thereby $\ell_Q, \mat{\Lambda}_Q$. For notational convenience, we refrain from introducing a scaling factor, but let it be absorbed into the covariance function.) With this (conjugate) prior, the posterior follows analytically by introducing the following statistics of the sampled trajectory $\vect{x}_{1:T}$:
\begin{align}
  \matb{\Phi} &= \sum_{t=1}^T \vectb{\zeta}_{t} \vectb{\zeta}_{t}\Transp, &
  \matb{\Psi} &= \sum_{t=1}^T \vectb{\zeta}_{t} \vect{z}_{t}\Transp, &
  \matb{\Sigma} &= \sum_{t=1}^T \vect{z}_{t} \vect{z}_{t}\Transp,
\end{align}
where $\vectb{\zeta}_t = \vect{x}_{t+1}$ and $\vect{z}_t = \big[\phi^{(1)}(\vect{x}_t) \dots \phi^{(m)}(\vect{x}_t) \big]\Transp$. Using the Markov property of the SSM, it is possible to write the conditional distribution for $\mat{Q}$ as \cite[Eq.~(42)]{WSL+:2012}:
\begin{multline} \label{eq:Qdist}
  p(\mat{Q}\mid \vect{x}_{1:T},\vect{y}_{1:T}) = \\ 
  \mathcal{IW}(\mat{Q}\mid T+\ell_Q,\matb{\Lambda}_Q + \left(\matb{\Phi} - \matb{\Psi}(\matb{\Sigma} + \mat{V})^{-1}\matb{\Psi}\Transp\right)).
\end{multline}
Given the prior \eqref{eq:A_prior}, $\mat{A}$ can now to be sampled from \cite[Eq.~(43)]{WSL+:2012}:
\begin{multline}\label{eq:Adist}
  p(\mat{A}\mid \mat{Q}, \vect{x}_{1:T}, \vect{y}_{1:T}) = \\ 
  \mathcal{MN}(\mat{A} \mid \matb{\Psi}(\matb{\Sigma} + \mat{V})^{-1},\mat{Q},(\matb{\Sigma} + \mat{V})^{-1}).
\end{multline}

\subsection{Marginalizing the Hyperparameters}
\label{sec:sample_hyperparams}
Concerning the sampling of the hyperparameters $\vectb{\theta}$, we note that we can easily evaluate the conditional distribution $p(\vectb{\theta} \mid \vect{x}_{1:T},\mat{Q},\mat{A})$ up to proportionality as
\begin{multline} \label{eq:hyp_post}
  p(\vectb{\theta} \mid \vect{x}_{1:T},\mat{Q},\mat{A}) \propto \\
  p(\vectb{\theta}) \, p(\mat{Q}\mid \vect{x}_{1:T},\mat{Q},\vectb{\theta}) \, p(\mat{A}\mid \vect{x}_{1:T},\mat{Q},\mat{A},\vectb{\theta}).
\end{multline}
To utilize this, we suggest to sample the hyperparameters by using a Metropolis--Hastings (MH) step, resulting in a so-called Metropolis-within-Gibbs procedure.

\section{THEORETICAL RESULTS}
\label{sec:convergence}

Our model \eqref{eq:rrgpssm} and learning Algorithm~\ref{alg:gibbs} inherits certain well-defined properties from the reduced-rank approximation and the presented sampling scheme. In the first theorem, we consider the convergence of a series expansion approximation to the GP-SSM with an increasing number $m$ of basis functions. As in \cite{SS:2014}, we only provide the convergence results for a rectangular domain with Dirichlet boundary conditions, but the result could easily be extended to a more general case. Proofs for all theorems are included in the supplementary material.

\begin{theorem}\label{thm:approx}
  The probabilistic model implied by the dynamic and measurement models of the approximate GP-SSM convergences in distribution to the exact GP-SSM, when the size of the domain $\Omega$ and the number of basis functions $m$ tends to infinity.
\end{theorem}

The above theorem means that in the limit any probabilistic inference in the approximate model will be equivalent to inference on the exact model, because the prior and likelihood models become equivalent. The benefit of considering the $m$-rank model instead of a standard GP, is the following:

\begin{theorem}\label{thm:comp_load}
  Provided the rank-reduced approximation, the computational load scales as $\mathcal{O}(m^2 T)$ as opposed to $\mathcal{O}(T^3)$.
\end{theorem}

Furthermore, the proposed learning procedure enjoys sound theoretical properties:

\begin{theorem}\label{thm:inv_dist}
  Assume that the support of the proposal in the MH algorithm covers the support of the posterior $p(\vectb{\theta} \mid \vect{x}_{1:T}, \mat{Q}, \mat{A},\vect{y}_{1:T})$, and $N \geq 2$ in Algorithm~\ref{alg:CPFAS}. Then the invariant distribution of Algorithm~\ref{alg:gibbs} is $p(\vect{x}_{1:T}, \mat{Q}, \mat{A}, \vectb{\theta} \mid \vect{y}_{1:T})$.
\end{theorem}

Hence, Theorem~\ref{thm:inv_dist} guarantees that our learning procedure indeed is sampling from the distribution we expect it to, even when a \emph{finite} number of particles $N \geq 2$ is used in the Monte Carlo based Algorithm~\ref{alg:CPFAS}. It is also possible to prove uniform ergodicity for Algorithm~\ref{alg:gibbs}, as such a result exists for Algorithm~\ref{alg:CPFAS} \citep{LJS:2014}.

\section{NUMERICAL EXAMPLES}
\label{sec:examples}

\begin{table*}[!t]
	\caption{Results for synthetic and real-data numerical examples.}
	\label{tab:ex}
	\footnotesize\scshape
	\begin{tabularx}{\textwidth}{@{\extracolsep{\fill}}l l c c c c l} 
		\toprule
		\multicolumn{2}{l}{Data / Method} & RMSE & LL & Train time [min] & Test time [s] & Comments \\
		\midrule
		\multicolumn{2}{l}{\emph{Synthetic data:}} \\
		& PMCMC GP-SSM \cite{FLS+:2013} 		& $\bm{1.12}$ 	& $-1.57$ 		& 547 		 & $420$ 		& As reported by \cite{FCR:2014}. \\
		& Variational GP-SSM \cite{FCR:2014}	& ${1.15}$ 		& $-1.61$ 		& $2.1$ 	 & $\bm{0.14}$ 	& As reported by \cite{FCR:2014}. \\
		& Reduced-rank GP-SSM 					& $\bm{1.10}$ 	& $\bm{-1.52}$ 	& $\bm{0.7}$ & $\bm{0.18}$ 		& Average over 10 runs.\\
		\multicolumn{2}{l}{\emph{Damper modeling:}} \\
		& Linear OE model (4th order) 		&$27.1$ 		& N/A \\
		& Hammerstein--Wiener (4th order) 	&$27.0$ 		& N/A \\
		& NARX (3rd order, wavelet network) &$24.5$	 		& N/A \\
		& NARX (3rd order, Tree partition)  &$19.3$	 		& N/A \\
		& NARX (3rd order, sigmoid network) &$\bm{8.24}$ 	& N/A \\
		& Reduced-rank GP-SSM 				&$\bm{8.17}$ 	& $\bm{-3.71}$ \\
		\multicolumn{2}{l}{\emph{Energy forecasting:}} \\
		& Static GP & $27.7$ & $-2.54$ \\
		& Reduced-rank GP-SSM & $\bm{21.8}$ & $\bm{-2.41}$ \\
		\bottomrule
	\end{tabularx}
	\vspace*{-1.5em}
\end{table*}

In this section, we will demonstrate and evaluate our contribution, the model \eqref{eq:rrgpssm} and the associated learning Algorithm~\ref{alg:gibbs}, using four numerical examples. We will demonstrate and evaluate the proposed method (including the convergence of the learning algorithm) on two synthetic examples and two real-world datasets, as well as making a comparison with other methods.

In all examples, we separate the data set into training data $\vect{y}^\textrm{t}$ and evaluation data $\vect{y}^\textrm{e}$. To evaluate the performance quantitatively, we  compare the estimated data $\widehat{\vect{y}}$ to the true data $\vect{y}^\textrm{e}$ using the root mean square error (RMSE) and the mean log likelihood (LL):
\begin{equation}
  \textrm{RMSE} = \sqrt{\frac{1}{T_\textrm{e}}\sum_{t=1}^{T_e}\left|\widehat{\vect{y}}_t-\vect{y}_t^\textrm{e}\right|^2}
\end{equation}
and
\begin{equation}
  \textrm{LL} = \frac{1}{T_\textrm{e}}\sum_{t=1}^{T_e}\log\N(\vect{y}_t^\textrm{e}\mid\mathbb{E}[\widehat{\vect{y}}_t],\mathbb{V}[\widehat{\vect{y}}_t]).
\end{equation}
The source code for all examples is available via the first authors homepage.

\subsection{Synthetic Data}
\label{sec:ex1}
As a proof-of-concept already presented in Figure~\ref{fig:simple}, we have $T = 500$ data points from the model
\begin{equation}
  x_{t+1} = \mathrm{tanh}(2x_t) + w_t, \qquad
  y_t = x_t + e_t,
\end{equation}
where $e_t\sim \N(0,0.1)$ and $w_t \sim \mathcal{N}(0,0.1)$. We inferred $f$ and $Q$, using a GP with the exponentiated quadratic (squared exponential, parametrized as in \citealt{RW:2006}) covariance function with unknown hyperparameters, and $Q \sim \mathcal{IW}(10,1)$ as priors. In this one-dimensional case ($x \in [-L, L], L = 4$), the eigenvalues and eigenfunctions are $\lambda_j = (\pi j/(2 L))^2$ and $\phi^{(j)}(x) = 1/\sqrt{L} \sin(\pi j (x + L)/(2 L))$. The spectral density corresponding to the covariance function is $S_\theta(\omega) = \sigma^2 \sqrt{2\pi\ell^2}\exp(-\omega^2 \ell^2/2)$.

The posterior estimate of the learned model is shown in Figure~\ref{fig:simple}, together with the samples of the basis function weights $f^{(j)}$. The variance of the posterior distribution of $f$ increases in the regimes where the data is not exciting the model.

As a second example, we repeat the numerical benchmark example on synthetic data from \cite{FCR:2014}: A one-dimensional state space model $x_{t+1} = x_t+1 + w_t$, if $x_t < 4$, and $x_{t+1} = -4x_t + 21$, if $x_t \geq 4$ with known measurement equation $y_t = x_t + e_t$, and noise distributed as $w_t \sim \N(0,1)$ and $e_t \sim \N(0,1)$. The model is learned from $T = 500$ data points, and evaluated with $T_\textrm{e} = 10^5$ data points. As in \cite{FCR:2014}, a Mat\'{e}rn covariance function is used (see, \eg, Section 4.2.1 of \citealt{RW:2006} for details, including its spectral density). The results for our model with $K = 200$ MCMC iterations and $m = 20$ basis functions are provided in Table~\ref{tab:ex}.

\begin{figure}[t!]
	\centering%
	\includegraphics{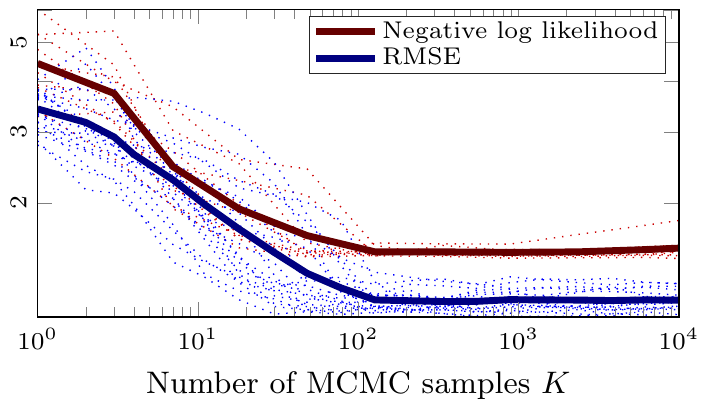}
	\caption{The (negative) log likelihood and RMSE for the second synthetic example, as a function of number of MCMC samples $K$, averaged (solid lines) over 10 runs (dotted lines).}
	\vspace*{-1em}
	\label{fig:simple-convergence}
\end{figure}

We also re-state two results from \cite{FCR:2014}: The GP-SSM method by \cite{FLS+:2013} (which also uses particle MCMC for learning) and the variational GP-SSM by \cite{FCR:2014}. Due to the compact writing in \cite{FLS+:2013,FCR:2014}, we have not been able to reproduce the results, but to make the comparison as fair as possible, we average our results over 10 runs (with different realizations of the training data). Our method was evaluated using the provided Matlab implementation on a standard desktop computer\footnote{Intel~i7-4600 2.1~GHz CPU.}.

The choice to use only $K = 200$ iterations of the learning algorithm is motivated by Figure~\ref{fig:simple-convergence}, illustrating the `model quality' (in terms of log likelihood and RMSE) as a function of $K$: It is clear from Figure~\ref{fig:simple-convergence} that the model quality is of the same magnitude after a few hundred samples and after $10\,000$ samples. As we know the sampler converges to the right distribution in the limit $K \rightarrow \infty$, this indicates that the sampler converges already after a few hundred samples for this example. This is most likely thanks to the linear-in-parameter structure of the reduced-rank GP, allowing for the efficient Gibbs updates (\ref{eq:Qdist}--\ref{eq:Adist}).

There is an advantage for our proposed reduced-rank GP-SSM in terms of LL, but considering the stochastic elements involved in the experiment, the different RMSE performance results are hardly outside the error bounds. Regarding the computational load, however, there is a substantial advantage for our proposed method, enjoying a training time less only a third of the one by the variational GP-SSM, which in turn outperforms the method by \cite{FLS+:2013}.

\subsection{Nonlinear Modeling of a Magneto-Rheological Fluid Damper}
\label{sec:mr}
We also compare our proposed method to state-of-the-art conventional system identification methods \citep{Ljung:1999}. The problem is the modeling of input--output behavior of a magneto-rheological fluid damper, introduced by \cite{WSC+:2009} and used as a case study in the System Identification Toolbox for Mathworks Matlab \citep{Mathworks:2015}. The data consists of $3\,499$ data points, of which $2\,000$ are used for training and the remaining for evaluation, shown in Figure~\ref{fig:damper}. The data exhibits some non-trivial dynamics, and as the $T=2\,000$ data points probably not contain enough information to determine the system uniquely, a certain amount of uncertainty is present in the posterior. This is thus an interesting and realistic problem for a Bayesian method, as it possibly can provide useful information about the posterior uncertainty, not captured in classical maximum likelihood methods for system identification.

We learn a three-dimensional model:
\begin{subequations}
\begin{align}\label{eq:ex2:mod}
  \vect{x}_{t+1} &= \vect{f}_{x}(\vect{x}_t) + \vect{f}_{u}(u_t) + \vect{w}_t, \\
  y_t &= [0~0~1]\vect{x}_t + e_t
\end{align}
\end{subequations}
where $\vect{x}_t \in \Reals^3$, $e_t \sim \N(0,5)$, and $\vect{w}_t \sim \N(\vect{0},\mat{Q})$ with $\mat{Q}$ unknown. We assume a GP prior with an exponentiated quadratic covariance function, with separate length-scales for each dimension. We use $m = 7^3 = 343$ basis functions\footnote{Explicit expression for the basis functions in the multidimensional case is found in the supplementary material.} for $\vect{f}_x$ and $8$ for $\vect{f}_u$, which in total gives $1\,037$ basis function weights $f^{(j)}$ and $5$ hyperparameters $\vectb{\theta}$ to sample.

The learned model was used to simulate a distribution of the output for the test data, plotted in Figure~\ref{fig:damper}. Note how the variance of the prediction changes in different regimes of the plot, quantifying the uncertainty in the posterior belief. The resulting output is also evaluated quantitatively in Table~\ref{tab:ex}, together with five state-of-the-art maximum likelihood methods, and our proposed method performs on par with the best of these. The learning algorithm took about two hours to run on a standard desktop computer.

The assumed model with known linear $g$ and additive form $\vect{f}_x + \vect{f}_u$ could be replaced by an even more general structure, but this choice seems to give a sensible trade-off between structure (reducing computational load) and flexibility (increasing computational load) for this particular problem. Our proposed Bayesian method does indeed appear as a realistic alternative to the maximum likelihood methods, without any more problem-specific tailoring than the rather natural model assumption \eqref{eq:ex2:mod}.

\begin{figure}[t!]
  \hspace*{\fill}
  \begin{subfigure}[b]{\columnwidth}
    \centering%
    \includegraphics{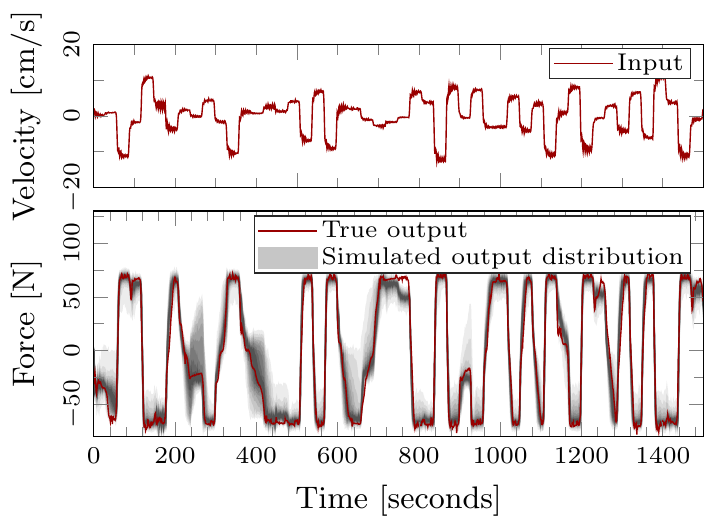}
    \caption{Fluid damper results}
    \label{fig:damper}
  \end{subfigure}
  \hspace*{\fill}
  \begin{subfigure}[b]{\columnwidth}
	\centering%
    \includegraphics{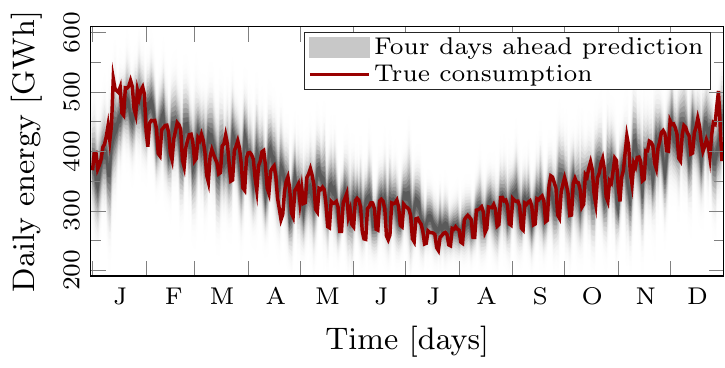}
    \caption{Electricity consumption example}
    \label{fig:energy}
  \end{subfigure}
  \hspace*{\fill}
  \vspace*{-1em}
  \caption{Data (red) and predicted distributions (gray) for the real-data examples. It is interesting to note how the variance in the prediction changes between different regimes in the plots.}
  \vspace*{-.5em}
  \label{fig:real}
\end{figure}

\subsection{Energy Consumption Forecasting}
\label{sec:energy}
As a fourth example, we consider the problem of forecasting the daily energy consumption in Sweden  
\footnote{Data from Svenska Kraftn\"{a}t, available: \url{http://www.svk.se/aktorsportalen/elmarknad/statistik/}.} 
four days in advance. The daily data from 2013 was used for training, and the data from 2014 for evaluation. The time-series was modeled as an autonomous dynamical system (driven only by noise), and a three dimensional reduced-rank GP-SSM was trained for this, with all functions and parameters unknown. To obtain the forecast, the model was used inside a particle filter to find the state distribution, and the four step ahead prediction density was computed. The data and the predictions are shown in Figure~\ref{fig:energy}.

As a sanity check, we compare to a standard GP, not explicitly accounting for any dynamics in the time-series. The standard GP was trained to the mapping from $y_t$ to $y_{t+4}$, and the performance is evaluated in Table~\ref{tab:ex}. From Table~\ref{tab:ex}, the gain of encoding dynamical behavior in the model is clear.

\section{DISCUSSION}
\label{sec:discussion}
\subsection{Tuning}
For a successful application of the proposed algorithm, there are a few algorithm-specific parameters for the user to choose: The number of basis functions $m$ and the number of particles $N$ in PGAS. A large number of basis functions $m$ makes the model more flexible and the reduced-rank approximation `closer' to a non-approximated GP, but it also increases the computational load. With a smooth covariance function~$\kappa$, the prior is in practice $f^\ji \approx 0$  for moderate $j$, and $m$ can be chosen fairly small (as a rule of thumb, say, 6--15 per dimension) without making a too crude approximation. In our experience, the number of particles $N$ in PGAS can be chooses fairly small (say, $20$), without affecting the mixing properties of the Markov chain heavily. This is in accordance to what has been reported in the literature by \citet{LJS:2014}.

\subsection{Properties of the Proposed Model}
We have proposed to use the reduced-rank approximation of GPs by \cite{SS:2014} within a state space model, to obtain a GP-SSM which efficiently can be learned using a PMCMC algorithm. As discussed in Section~\ref{sec:learning} and studied using numerical examples in Section~\ref{sec:examples}, the linear-in-the-parameter structure of the reduced-rank GP-SSM allows for a computationally efficient learning algorithm. However, the question if a similar performance could be obtained with another GP approximation method or another learning scheme arises naturally.

Other GP approximation methods, for example pseudo-inputs, would most likely not allow for such efficient learning as the reduced-rank approximation does; unless closed-form Gibbs updates are available (requiring a linear-in-the-parameter structure, or similar), the parameter learning would have to resort to Metropolis--Hastings, which most likely would give a significantly slower learning procedure. For many GP approximation methods it is also more natural to find a point estimate of the parameters (the inducing points, for example) using, for example, EM, rather than inferring the parameter posterior, as is the case in this paper.

The learning algorithm, on the other hand, could probably be replaced by some other method also inferring (at least approximately) the posterior distribution of the parameters, such as SMC$^2$ \citep{CJP:2013} or a variational method. However, to maintain efficiency, the method has to utilize the linear-in-the-parameter structure of the model to reach a computational load competitive with our proposed scheme. Such an alternative (however only inferring MAP estimate of the sought quantities) could possibly be the method by \cite{KSS:2014}.

\subsection{Conclusions}
We have proposed the reduced-rank GP-SSM \eqref{eq:rrgpssm}, and provided theoretical support for convergence towards the full GP-SSM. We have also proposed a theoretically sound MCMC-based learning algorithm (including the hyperparameters) utilizing the structure of the model efficiently.

By demonstration on several examples, the computational load and the modeling capabilities of our approach have been proven to be competitive. The computational load of the learning is even less than in the variational sparse GP solution provided by \cite{FCR:2014}, and the performance in challenging input--output modeling is on par with well-established state-of-the-art maximum likelihood methods.

\subsection{Possible Extensions and Further Work}
A natural extension for applications where some domain knowledge is present, is to let the model include some functions with an \textit{a~priori} known parametrization. The handling of such models in the learning algorithm should be feasible, as it is already known how to use PGAS for such models \citep{LJS:2014}. Further, the assumptions of the $\mathcal{IW}$ prior of $\mat{Q}$ \eqref{eq:Qprior} are possible to circumvent by using, for example, MH, at the cost of an increased computational load. The same holds true for the Gaussian noise assumption in \eqref{eq:rrgpssm}.

Another direction for further work is to adapt the process to be able to sequentially learn and improve the model when data is added in batches, by formulating the previously learned model as the prior to the next iteration of the learning. This could probably be useful in, for example, reinforcement learning, along the lines of \cite{DFR:2015}. 

In the engineering literature, dynamical systems are mostly defined in discrete time. An interesting approach to model the continuous-time counterpart using Gaussian processes is presented by \cite{RBO:2013}. A development of the reduced-rank GP-SSM to continuous time dynamical models using stochastic Runge--Kutta methods would be of great interest for further research.

\newpage
\subsection*{References}

\begingroup
\renewcommand{\section}[2]{}%
\bibliographystyle{abbrunsrtnat}
\bibliography{references}
\endgroup

\clearpage
\normalsize\onecolumn\thispagestyle{empty}

\setcounter{section}{0}
\section*{\huge Supplementary material}

\setcounter{section}{0}

This is the supplementary material for `Computationally Efficient Bayesian Learning of Gaussian Process State Space Models' by Svensson, Solin, S\"{a}rkk\"{a} and Sch\"{o}n, presented at AISTATS~2016. The references in this document point to the bibliography in the article.

\section{Proofs}

\begin{proof}[Proof of Theorem~\ref{thm:approx}]
  Let us start by considering the GP approximation to $\vect{f}(\vect{x})$, $\vect{x} \in [-L_1,L_1] \times \cdots \times [-L_d,L_d]$. By Theorem~4.4 of \cite{SS:2014}, when domain size $\inf_i L_i \to \infty$ and the number of basis functions $m \to \infty$, the approximate covariance function $\vectb{\kappa}_m(\vect{x},\vect{x}')$ converges point-wise to $\vectb{\kappa}(\vect{x},\vect{x}')$. As the prior means of the exact and approximate GPs are both zero, the means thus converge as well. By similar argument as is used in the proof of Theorem~2.2 in \cite{SP:2014} it follows that the posterior mean and covariance functions will converge point-wise as well.

Now, consider the random variables defined by
\begin{align}
  \vect{x}_{t+1} &= \vect{f}(\vect{x}_t) + \vect{w}_t, \\
  \hat{\vect{x}}_{t+1} &= \vect{f}_m(\vect{x}_t) + \vect{w}_t,
\end{align}
where $\vect{f}_m$ is an $m$-term series expansion approximation to the GP. It now follows that for any fixed $\vect{x}_t$ the mean and covariance of $\vect{x}_{t+1}$ and $\hat{\vect{x}}_{t+1}$ coincide when $L_i, m \to \infty$. However, because these random variables are Gaussian, the first two moments determine the whole distribution and hence we can conclude that $\hat{\vect{x}}_{t+1} \to \vect{x}_{t+1}$ in distribution.

For the measurement model we can similarly consider the random variables
\begin{align}
  \vect{y}_t &= \vect{g}(\vect{x}_t) + \vect{e}_t, \\
  \hat{\vect{y}}_t &= \vect{g}_m(\vect{x}_t) + \vect{e}_t,
\end{align}
With similar argument as above, we can conclude that the approximation converges in distribution.
\end{proof}

\begin{proof}[Proof of Theorem~\ref{thm:comp_load}]
  Provided the reduced-rank approximation of the Gram matrix, the reduction in the computational load directly follows from application of the matrix inversion lemma.
\end{proof}

\begin{proof}[Proof of Theorem~\ref{thm:inv_dist}]
  Using fundamental properties of the Gibbs sampler (see, \eg, \cite{Tierney:1994}), the claim holds if all steps of Algorithm~\ref{alg:gibbs} are leaving the right conditional probability density invariant. Step~\ref{alg:gibbs:x} is justified by \cite{LJS:2014} (even for a finite $N$), and step~\ref{alg:gibbs:q}--\ref{alg:gibbs:a} by \cite{WSL+:2012}. Further, step~\ref{alg:gibbs:eta} can be seen as a Metropolis-within-Gibbs procedure \citep{Tierney:1994}.
\end{proof}

\section{Details on Matrix Normal and Inverse Wishart distributions}

As presented in the article, the matrix normal inverse Wishart (MNIW) distribution is the conjugate prior for state space models linear in its parameters $\mat{A} \in \Reals^{n \times m}$ and $\mat{Q} \in \Reals^{n \times n_x}$ \cite{WSL+:2012}. The MNIW distribution can be written as $\mathcal{MN}(\mat{A},\mat{Q}\mid M,\mat{V},\ell,\matb{\Lambda}) = \mathcal{MN}(\mat{A}\mid M,\mat{Q},\mat{V})\times\mathcal{IW}(\mat{Q}\mid\ell,\matb{\Lambda})$, where each part is defined as follows:
\begin{itemize}
	\item The pdf for the Inverse Wishart distribution with $\ell$ degrees of freedom and positive definite scale matrix $\matb{\Lambda} \in \Reals^{n \times n}$:
	\begin{equation}
	\mathcal{IW}(\mat{Q}\mid\ell,\matb{\Lambda}) = \frac{|\matb{\Lambda}|^{\ell/2}|\mat{Q}|^{-(n+\ell+1)/2}}{2^{\ell n/2}\Gamma_{n}(\ell/2)}\exp\left(-\frac{1}{2}\mathrm{tr}(\mat{Q}^{-1}\matb{\Lambda})\right)\label{eq:IWpdf}
	\end{equation}
	with $\Gamma_{n}(\cdot)$ being the multivariate gamma function.
	
	\item The pdf for the Matrix Normal distribution with mean $\mat{M} \in \Reals^{n \times m}$, right covariance $\mat{Q} \in \Reals^{n \times n}$ and left precision $\mat{V} \in \Reals^{m \times m}$:
	{\begin{equation}
	\mathcal{MN}(\mat{A}\mid \mat{M},\mat{Q},\mat{V}) = \tfrac{|\mat{V}|^{n/2}}{(2\pi)^{nm}|\mat{Q}|^{m/2}}\exp\left(-\frac{1}{2}\mathrm{tr}\left((\mat{A}-\mat{M})\Transp\mat{Q}^{-1}(\mat{A}-\mat{M})\mat{V}\right)\right) \label{eq:MNpdf}
	\end{equation}}%
\end{itemize}
To sample from the MN distribution, one may sample a matrix $\mat{X} \in \Reals^{n \times m}$ of i.i.d. $\N(0,1)$ random variables, and obtain $\mat{A}$ as $\mat{A} = \mat{M} + \mathrm{chol}(\mathcal{\mat{Q}})\,\mat{X}\,\mathrm{chol}(\mat{V})$, where $\mathrm{chol}$ denotes the Cholesky factor ($\mat{V} = \mathrm{chol}(\mat{V}) \, \mathrm{chol}(\mat{V})\Transp$).

\section{Eigenfunctions for Multi-Dimensional Spaces}

The eigenfunctions for a $d$-dimensional space with a rectangular domain $[-L_1,L_1] \times \dots \times [-L_d,L_d]$, used in Example~\ref{sec:mr} and Example~\ref{sec:energy}, are on the form
\begin{align}
	&\phi^{(j_1,\dots,j_d)}(x) = \prod_{k=1}^{d}\frac{1}{\sqrt{L_k}} \sin\left(\frac{\pi j_k (x_k+L_k)}{2L_k}\right) \qquad \text{with} \quad \lambda_{j_1,\dots,j_d} = \sum_{k=1}^{d} \left(\frac{\pi j_k}{2L_k}\right)^2.
\end{align}
Note how this for $d=1$ reduces to the univariate case presented in Section~\ref{sec:ex1}. For further details we refer to Section 4.2 in \citet{SS:2014}.
	
\section{Provided Matlab Software}

The following Matlab files are available via the first authors homepage:

{
	\footnotesize
\begin{tabularx}{\textwidth}{l l p{3.1cm}} 
	\toprule
	\bf{File} & \bf{Use} & \bf{Comments} \\
	\midrule
	\texttt{synthetic\_example\_1.m} & First synthetic example (including Figure~\ref{fig:simple}) & \\
	\texttt{synthetic\_example\_2.m} & Second synthetic example & \\
	\texttt{damper.m} & MR damper example & {\scriptsize For other results, see \cite{Mathworks:2015} }\\
	\texttt{energy\_forecast.m} & Energy consumption forecasting example & \\
	\texttt{iwishpdf.m} & Implements \eqref{eq:IWpdf} & \\
	\texttt{mvnpdf\_log.m} & Logarithm of normal distribution pdf & \\
	\texttt{systematic\_resampling.m} & Systematic resampling (Step~\ref{alg:CPFAS:resampling}, Algorithm~\ref{alg:CPFAS}) & \\
	\bottomrule
\end{tabularx}
}

All files are published under the GPL license.

\end{document}